\documentclass{ws}

\makeatletter
\let\trimmarks\relax
\let\cropmarks\relax

\makeatother

\usepackage{graphicx}
\usepackage{amsmath,amssymb,amsfonts}
\usepackage{newtxtext, newtxmath}
\usepackage{algorithm}
\usepackage{algpseudocode}

\title{ \LARGE Decoding Algorithms for Symbol-Error Correction in MDS Array Codes via Superregular Matrices}

\author{ \Large
	D\'ebora Beatriz Claro Zanitti$^{1}$,\quad
	Isabella Silva Teixeira$^{1}$,\\
	Carina Alves$^{2}$,\quad
	Ivan Aritz Aldaya Garde$^{1}$,\\
	Cintya Wink de Oliveira Benedito$^{1}$
	\footnote{
		$^{1}$School of Engineering, S\~ao Paulo State University (UNESP), 
		S\~ao Jo\~ao da Boa Vista, S\~ao Paulo, Brazil. \\
		$^{2}$Department of Mathematics, S\~ao Paulo State University (UNESP), 
		Rio Claro, S\~ao Paulo, Brazil.\\
		E-mails: debora.zanitti@unesp.br,  isabella.s.teixeira@unesp.br, carina.alves@unesp.br, ivan.aldaya@unesp.br,  cintya.benedito@unesp.br} 
}

\begin{document}
	\maketitle
	
	\begin{abstract}
		Maximum distance separable (MDS) array codes constitute an important class of error-correcting codes due to their optimal distance properties and their relevance in distributed storage systems. In this paper, we investigate the construction and decoding of MDS array codes over $\mathbb{F}_q^b$ based on superregular matrices, with a particular emphasis on superregular Vandermonde and Cauchy matrices. We propose decoding algorithms for $[n,k,d]$ MDS array codes, where $n=m+k$ and $d=m+1$, capable of correcting symbol errors without prior knowledge of their locations. Unlike existing approaches restricted to specific parameter settings, the proposed algorithms apply to general configurations and rely on algebraic relations that do not follow from straightforward extensions of previous methods. Specifically, these algorithms correct one symbol error for $m \geq 2$ and two symbol errors for $m \geq 4$. For the two-error case, the decoding procedure admits a simplified form when Vandermonde superregular matrices are employed, thereby reducing computational complexity.We also analyze the algebraic structure of the three-symbol-error case, focusing on the most involved configuration in which all errors occur in information symbols, and we discuss how the method may be extended to the general case. These algorithms are computationally efficient for moderate parameter sizes, as they rely on structured algebraic operations over $\mathbb{F}_q^b$ and the solution of small linear systems, making them suitable for distributed storage applications. From an application perspective, the proposed approach provides a flexible alternative to RAID~6 schemes. Unlike RAID~6, which is typically limited to two parity disks and often requires prior knowledge of error locations, our construction supports general configurations and enables the correction of multiple symbol errors without location information, at the cost of increased algebraic complexity.

	\end{abstract}
	
	
	
	\section{Introduction}\label{sec1}
	
	Error-correcting codes are designed to ensure reliable transmission and storage of information in the presence of noise and interference. Depending on the channel characteristics, errors may affect isolated symbols or occur in bursts, impacting multiple consecutive symbols \cite{huffman,handbook}. The basic idea of an error-correcting code is to encode information by adding redundancy in an organized way, allowing the receiver to detect and correct errors \cite{channel}. 
	
	An important property to analyze in coding theory is the minimum code distance. This metric is associated with the error-correcting ability of the code; specifically, a greater minimum distance reduces the probability of incorrect correction to the nearest codeword. Therefore, constructing codes with the maximum possible minimum distance is desirable to achieve effective error correction. Codes that achieve this property of maximum distance separation, that is, codes whose minimum Hamming distance attains the Singleton bound, are known as maximum distance separable (MDS) codes \cite{codigos}.

	Array codes are two-dimensional error-correcting codes widely used in storage systems due to their flexibility and efficiency in handling symbol-level errors \cite{handbook}. They can be constructed from simple parity-based schemes \cite{matriciais}, or algebraically, with symbols taken from a finite field, and can thus be viewed as linear codes with parameters $[n,k,d]$, where $n$ is the code length, $k$ is the dimension, and $d$ is the minimum distance. This perspective has been extensively used in the study of MDS array codes and related constructions \cite{fang}. This work focuses on algebraic constructions of MDS array codes and their decoding procedures, building on classical approaches such as Blaum--Roth codes and their variants \cite{Blaum,sara}. Recent research has investigated MDS array codes in the context of distributed storage systems, with emphasis on efficient repair schemes, sub-packetization, and locally repairable structures \cite{mds_subpacketization,fang,bargexplicit,bargcooperative}. In this setting, constructions based on Vandermonde and Cauchy matrices have also been explored for distributed storage systems \cite{sistema_armazenamento}.
	
	Distributed storage systems are commonly implemented using Redundant Array of Independent Disks (RAID) architectures \cite{Gibson, blaumevenodd}, where redundancy is introduced to ensure fault tolerance. In particular, RAID~6 schemes typically rely on constructions such as Reed--Solomon or EVENODD codes and can correct up to two disk failures (erasures) when their locations are known  \cite{marioblaum,Ping}. This limitation motivates the study of decoding algorithms capable of correcting symbol errors at unknown locations. An algebraic approach to constructing a system equivalent to RAID~6, using different techniques from those considered in this work, is presented in \cite{moussa2018raid}.
	
	Several studies have proposed decoding algorithms for MDS array codes under specific assumptions or restricted parameter settings \cite{blaumevenodd,BlaumMDS,Blaum,sara,sara2}. In particular, the work in \cite{sara} presents decoding algorithms for MDS array codes with fixed parameters, which limits the number of information blocks and, consequently, the minimum distance and error-correction capability. Increasing the error-correction capability from one to multiple errors, particularly in the absence of error-location information, is a substantially more involved problem. Extending constructions and decoding procedures from specific parameter instances to general code families widely recognized as a highly non-trivial task. Such extensions require handling complex algebraic dependencies between symbols and cannot, in general, be obtained by straightforward modifications of existing results.
	
	The main contribution of this paper is the development of explicit algebraic decoding procedures for MDS array codes with parameters $[m+k,k,m+1]$ constructed from superregular matrices. The proposed algorithms correct one symbol error for $m\geq 2$ and two symbol errors for $m\geq 4$ without assuming prior knowledge of the error locations. In contrast to approaches designed for specific code constructions, the proposed framework applies to general configurations. As a consequence of the increased complexity of this setting, we also derive a simplified decoding procedure for the case in which the underlying superregular matrix is a Vandermonde matrix, leading to simplified decoding expressions and a reduction in the number of required algebraic operations. We further examine the three-symbol-error case to illustrate the algebraic structure underlying possible extensions to higher numbers of errors, together with a corresponding simplification when Vandermonde matrices are employed.

	From an application perspective, in distributed storage systems with $k$ information disks and $m$ parity disks, RAID~6 employs $m = 2$ parity disks and can correct up to two errors only when their locations are known. Our approach considers general parameter configurations (including cases $m > k$, $m < k$, and $m = k$) and enables symbol error correction without prior knowledge of error locations. Moreover, the ability to correct more than two symbol errors allows the proposed framework to exceed the error-correction capability of RAID~6, providing a more flexible and reliable solution, albeit at the cost of increased computational complexity. From a computational perspective, the decoding procedures rely on structured algebraic operations that can be efficiently implemented, making them suitable for practical applications in distributed storage systems.

	This paper is organized as follows. In Section~\ref{initialconcepts}, we introduce the fundamental concepts used throughout the paper. Section~\ref{construção} details the  construction of the MDS array codes used in this work. Next, Section~\ref{seção_decodificação} presents decoding algorithms for MDS array codes capable of correcting up to two symbol errors, together with a simplification for the two-error case when Vandermonde matrices are employed. In Subsection~\ref{seção_correção_três_erros},  these algorithms are further examined in the context of correcting three symbol errors. Finally, Section~\ref{conclusao} summarizes the main results and discusses directions for future research.

	\section{Initial Concepts}\label{initialconcepts}
	
	A \textbf{linear code} $\mathcal{C}$ is a $K$-dimension  subspace of $\mathbb{F}_q^N$, where $\mathbb{F}_q$ is a finite field with $q$ elements. We describe the  code $\mathcal{C}$ through the parameters $[N,K,D]$, where $N$ is the code length, $K$ is the dimension and $D$ is the minimum Hamming distance. For a linear code $\mathcal{C}$ we can calculate a bound to the number of codewords, which is given by the following theorem.
	
	\begin{theorem} \label{limitante_sing} \cite{codigos}(Singleton Bound) 
		If $\mathcal{C}$ is a linear code $[N,K,D]$ over $\mathbb{F}_q$ then the maximum number of possible codewords is $q^K$ and the Singleton bound says that  
		\begin{equation}\label{limi}
			q^K\leq q^{N-D+1},
		\end{equation} 
		or equivalently
		\begin{equation}\label{limi_d} 
			D\leq N-K+1.
		\end{equation}
	\end{theorem}
	A code in which equality is achieved on (\ref{limi_d}) is called a {\bf maximum distance separable - MDS} code,	which means that no code of length $N$ and minimum distance $D$ has more codewords than an MDS code with the same parameters.
	
	Now, if we take $b>0$ as a positive integer such that $b$ divides $K$ and $N$, then we can construct a \textbf{linear array code} $\mathcal{C}$ over $\mathbb{F}_q^b$, with parameters $[n,k,d]$, where $k = K/b$, $n = N/b$ and the minimum distance $d$ is calculated
	considering the code over $\mathbb{F}_q^b$. Such codes can be specified by their parity-check matrix $H$ of dimension $(n-k)b \times nb$,
	generating codewords of length $nb$, where each block of $b$ symbols is regarded as an element of  $\mathbb{F}_q^b$. In this setting, the error model corresponds to errors affecting entire symbols over $\mathbb{F}_q^b$, rather than errors in individual scalar components.
	The Singleton bound  is still valid in linear array codes, that is, $d\leq n-k+1$. And if equality is achieved we have an  {\bf MDS linear array code}.

	Next, we present some concepts that will be used in the construction of MDS array codes and in the decoding algorithms developed in this work. Let $\mathbb{F}_{q^b}$ be the extension field of $\mathbb{F}_q$ of degree $b$, and let $\alpha \in \mathbb{F}_{q^b}$ be a primitive element, that is, a root of a primitive polynomial $p(x) \in \mathbb{F}_q[x]$. Throughout the paper, we identify $\mathbb{F}_{q^b}$ with $\mathbb{F}_q^b$ as a vector space over $\mathbb{F}_q$ via a fixed basis, so that each symbol over $\mathbb{F}_{q^b}$ is represented as a block of $b$ components over $\mathbb{F}_q$. We denote by $M_{m \times n}(\mathbb{F}_q)$ the space of $m \times n$ matrices with entries in $\mathbb{F}_q$.
	
	\begin{definition} \label{definição_matriz_superregular}
		A matrix $A\in M_{m \times n}(\mathbb{F}_q)$ is called {\bf superregular}
		if every square submatrix of $A$ is nonsingular, that is, all its
		determinants are nonzero.
		
	\end{definition}

	\begin{definition}\label{bblock} 
		A matrix $A\in M_{mb \times nb}(\mathbb{F}_q)$ is said to be a
		{\bf superregular $b$-block matrix} if every square submatrix of $A$
		formed by full $b \times b$ blocks is nonsingular over $\mathbb{F}_q$.
		
	\end{definition}
	
	In the context of coding theory, superregular matrices with entries in the finite field $\mathbb{F}_q$ can be used to generate linear codes and linear array codes with good distance properties.
	
	\begin{proposition}\label{propmds}\cite{Roth}
		Let $\mathcal{C}$ be a linear code $[N,K,D]$ over a finite field $\mathbb{F}_q$. The following conditions are equivalent to $\mathcal{C}$ being  MDS:
		\begin{enumerate}
			
			\item Every set of $N-K$ columns of a $H$ parity-check matrix of $\mathcal{C}$ is linearly independent.
			
			\item Every set of $K$ columns of a generator matrix of $\mathcal{C}$ is linearly independent.
			
			\item The dual code $\mathcal{C}^\perp$ is MDS.
			
			\item The code $\mathcal{C}$ has a generator matrix in the systematic form  $G=(I_K | A)$, where $I_K$ is the identity matrix $K\times K$ and $A$ is a superregular matrix.
		\end{enumerate}
	\end{proposition}

	The Proposition \ref{propmds} can also be extended to linear array MDS codes $[n,k,d]$ by using superregular $b$-block matrices. Superregular matrices are not easy to find in general \cite{sara2}. We will consider particular examples of superregular matrices, specifically the Vandermonde and Cauchy matrices \cite{cauchy,roth2}. 
	
	Let $\alpha_1,..., \alpha_{n-k}$ be nonzero elements in   $\mathbb{F}_q$. A general form of a Vandermonde matrix $A\in\mathbb{F}_q^{k\times (n-k)}$ is given by
	
	\begin{equation} \label{vandermonde}
		A = \left [
		\begin{array}{c c c c}
			1 & 1 & ... & 1 \\
			\alpha_1 & \alpha_2 & \ldots & \alpha_{n-k}  \\
			\alpha_1^2 & \alpha_2^2 & \ldots & \alpha_{n-k}^2  \\
			\vdots & \vdots & \ddots & \vdots  \\
			\alpha_1^{k-1} & \alpha_2^{k-1} & \ldots & \alpha_{n-k}^{k-1} 
		\end{array}
		\right ].
	\end{equation}
	
	We emphasize that, over finite fields, a Vandermonde matrix is not automatically superregular; this property depends on the choice of evaluation points and must be verified.

	\begin{definition}\label{defcauchy}
		Let $x_i$ and $y_j$ elements in a field $\mathbb{F}$, for $i=1,\cdots,m$ and $j=1,\cdots,n$. A matrix $A$ of dimensions $m \times n$ is called \textbf{Cauchy matrix}, if and only if,
		\begin{equation}\label{Cauchy}
			A_{ij} = \frac{1}{x_i + y_j} ,
		\end{equation}
		where $x_i$ + $y_j\neq\textbf{0}$, $x_i$ are distinct for all $i=1,\cdots,m$ and $y_j$ are distinct for all $j=1,\cdots,n$.
	\end{definition}

	\begin{definition} Let $p(x) = x^{b} + p_{b-1}x^{b-1} + ... +  p_{1}x +  p_{0} \in \mathbb{F}_q[x]$ be a primitive polynomial. We can associate $p(x)$ with a matrix called \textbf{Frobenius companion matrix} $C$, which is a matrix that contains 1's (ones) in the sub-diagonal,  the coefficients of $-p(x)$ in the last column and all other elements being zero, as follows

		\begin{equation}\label{matrizC}
			C = \left [
			\begin{array}{c c c c c}
				0 & 0 & \ldots & 0 & -p_{0} \\
				1 & 0 & \ldots & 0 & -p_{1} \\
				\vdots & \ddots & \ddots & \vdots & \vdots \\
				0 & 0 & \ldots & 0 & -p_{b-2}\\
				0 & 0 & \ldots & 1 & -p_{b-1}
			\end{array}
			\right ].
		\end{equation}
	\end{definition}
	
	\begin{definition}  \label{definição_zech}  If $\alpha$ is a primitive element of a finite field $\mathbb{F}_q$, then the {\bf Zech logarithm} $Z$ of an integer $n$ with respect to the base $\alpha$ is defined by the equation
		\begin{equation} \label{zech}
			Z(n) = log_{\alpha}(1 + \alpha^n)
			\ \ \ \Longleftrightarrow \ \ \             \alpha^{Z(n)} = 1 + \alpha^{n},
		\end{equation}
		where the result of logarithm is limited to the field.
	\end{definition}
	
	Some properties of this logarithm can be found in \cite{Zech}. The sum of powers of a primitive element $\alpha$ can be done alternatively using the Zech logarithm. Let $\alpha^x$ and $\alpha^y$ be primitive elements of $\mathbb{F}_q$. If $z$ is the exponent of the sum $\alpha^x + \alpha^y$, that is, $\alpha^x + \alpha^y = \alpha^z$, then
	\begin{equation}\label{propriedade_zech}
		z= x + Z(y-x) = y + Z(x-y).
	\end{equation}
	
	\section{Constructions of MDS Array Codes} \label{construção}
	
	In this section we will present the encoding of MDS array codes based on \cite{Blaum,sara}. We start by defining  an isomorphism that takes a primitive element of a finite field into a Frobenius companion matrix.
	
	Let us consider $M_{m \times n}(\mathbb{F}_q)$  and $\alpha \in \mathbb{F}_{q^b}$, with $b$ a positive integer. There is a field isomorphism $\phi: \mathbb{F}_{q^b}\longrightarrow \mathbb{F}_q[C]$ defined by $\phi(\alpha)=C$, where $C\in M_{b\times b}(\mathbb{F}_q)$ is a Frobenius companion matrix. We can extend to a ring isomorphism defined by
	
	\begin{equation}\label{iso} \begin{array}{cccc}
			\psi: & M_{m \times n}(\mathbb{F}_{q^b}) &\longrightarrow& M_{mb \times nb}(\mathbb{F}_q)\\
			& A=[a_{ij}] & \longmapsto& \psi(A)=[\phi(a_{ij})]=[C^{\sigma(i,j)}]\end{array},
	\end{equation}
	where  $\sigma(i,j)$ is the power of $\alpha$ in each element $a_{ij}$ of $A$.
	
	Let $\mathcal{C}$ be a linear array code over $\mathbb{F}_q^b$ with parameters $[n,k,d]$. If $d$ satisfies the equality in (\ref{limi_d}), then $\mathcal{C}$ is an MDS array code. The following result provides a condition for $\mathcal{C}$ to be MDS and presents a parity-check matrix for $\mathcal{C}$.

	\begin{theorem}\label{teoSara} \cite{Roth,sara}
		If $A = [a_{ij}] \; \in \; M_{(n-k) \times k}(\mathbb{F}_{q^b})$ is a  superregular $b$-block matrix, then $H= \big[\psi(A) \ | \  I_{(n-k)b}\big]$ is a  parity-check matrix of a $[n, k, n- k + 1]$  MDS linear array code $\mathcal{C}$, where $\psi$ is the isomorphism given in (\ref{iso}).
	\end{theorem}

	Henceforth we assume $q=2$, that is,  binary codes and linear codes. Then, $\mathcal{C}$ is an MDS array code over $\mathbb{F}_2^b$. The  code parameters become $[n,k,n-k+1]=[m+k,k,m+1]$, where $m=n-k$. The parity-check matrix  is given by  
	\begin{equation}\label{equação_matrizH}	H = \big[\psi(A) \ | \  I_{mb}\big] = \left [
		\begin{array}{c c c c  | c}
			A_{11} & A_{12} & \ldots & A_{1k}  & \\
			A_{21} & A_{22} & \ldots & A_{2k}  & \\
			\vdots & \vdots & \ddots & \vdots &  I_{mb} \\
			A_{m1} & A_{m2} & \ldots & A_{mk} &
		\end{array}
		\right ], \end{equation}
	with $A_{ij}=C^{\sigma(i,j)}$, where $C\in M_{b\times b}(\mathbb{F}_2)$ is a Frobenius companion matrix, $\sigma(i,j)$ is the power $\alpha$ of each element $a_{ij}$ of a superregular matrix $A$ of dimension $m \times k$ and $I_{mb}$ is the identity matrix of dimension $mb$.

	\begin{example} \label{Exemplo: construção H}
		Let $p(x)= x^5 + x^3 + x^2 + x + 1 \; \in \; \mathbb{F}_2[x]$ be a primitive polynomial with primitive element $\alpha \in \mathbb{F}_{2^5}$ and let us construct an MDS array code $\mathcal{C}$ with parameters $[10,5,6]$ over $\mathbb{F}_{2}^5$, that is, with $m=k=b=5$. The superregular Vandermonde matrix and the Frobenius companion matrix associated to $p(x)$ are, respectively, given by
		\[
		A = \left [
		\begin{array}{c c c c c}
			1 & 1 & 1 & 1 & 1 \\
			\alpha & \alpha^2 & \alpha^3 & \alpha^4 & \alpha^5 \\
			\alpha^2 & \alpha^4 & \alpha^6 & \alpha^8 & \alpha^{10} \\
			\alpha^3 & \alpha^6 & \alpha^9 & \alpha^{12} & \alpha^{15} \\
			\alpha^4 & \alpha^8 & \alpha^{12} & \alpha^{16} & \alpha^{20} \\
		\end{array}
		\right ]  \ \ \mbox{and} \ \ \  C =
		\left [
		\begin{array}{c c c c c}
			0 & 0 & 0 & 0 & 1 \\
			1 & 0 & 0 & 0 & 1 \\
			0 & 1 & 0 & 0 & 1 \\
			0 & 0 & 1 & 0 & 1\\
			0 & 0 & 0 & 1 & 0
		\end{array}
		\right ].
		\]
		Applying the isomorphism  (\ref{iso}) we have that
		
		\[
		\psi(A) = \left[ 
		\begin{array}{c c c c c}
			I_{5} & I_{5} & I_{5}  & I_{5}  & I_{5}   \\
			C     & C^2   & C^3    & C^4    & C^5     \\
			C^2   & C^4   & C^6    & C^8    & C^{10}  \\
			C^3   & C^6   & C^9    & C^{12} & C^{15}  \\
			C^4   & C^8   & C^{12} & C^{16} & C^{20} 
		\end{array}
		\right ].
		\]
		Thus, according to (\ref{equação_matrizH}), the parity-check matrix for an MDS array code $\mathcal{C}$ is given by
		\[H = \left [
		\begin{array}{c c c c c | c}
			I_{5} & I_{5} & I_{5} & I_{5} & I_{5} &   \\
			C & C^2 & C^3 & C^4 & C^5  & \\
			C^2 & C^4 & C^6 & C^8 & C^{10} &  I_{25} \\
			C^3 & C^6 & C^9 & C^{12} & C^{15} &  \\
			C^4 & C^8 & C^{12} & C^{16} & C^{20} &
		\end{array}
		\right ]. \] \\
	\end{example}
	\section{Decoding of MDS Array Codes}\label{seção_decodificação}
	
	This section presents two decoding algorithms for MDS array codes with
	parameters $[n,k,n-k+1]=[m+k,k,m+1]$: one for correcting a single symbol
	error over $\mathbb{F}_2^b$ when $m \geq 2$, and another for correcting
	two symbol errors over $\mathbb{F}_2^b$ when $m \geq 4$, where
	$m = n - k$, without assuming knowledge of the error locations. A
	simplified version of the decoding algorithm for correcting two symbol errors is presented when the
	superregular matrix is chosen as a Vandermonde matrix.
	
	It is well known that linear block codes with parameters $[N,K,D]$ can correct up to $\left\lfloor\frac{D-1}{2}\right\rfloor$ symbol errors \cite{huffman}. Similarly, MDS array codes with parameters $[n,k,d]$, where $k = K/b$ and $n = N/b$, can correct up to $\left\lfloor\frac{d-1}{2}\right\rfloor$ symbol errors over $\mathbb{F}_2^b$. In Subsection~\ref{seção_correção_três_erros}, we present the progress made toward extending the proposed algorithms to
	handle three symbol errors, as well as the main obstacles encountered
	in this generalization. A simplified form is also presented for the case of three symbol errors when a Vandermonde matrix is used.

	Let $\mathcal{C}$ be an MDS array code over $\mathbb{F}_2^b$ constructed as in Section \ref{construção}. If $\textbf{c}=[\textbf{c}_1 \ \textbf{c}_2 \ \ldots \ \textbf{c}_n]$ is a transmitted codeword and $\textbf{v}=[\textbf{v}_1 \ \textbf{v}_2 \ \ldots \ \textbf{v}_n]$ is a received vector then $\textbf{e}= \textbf{v} - \textbf{c}$ is the error vector, where $\textbf{c}_j,\textbf{v}_j,\textbf{e}_j\in\mathbb{F}_2^b$, for $j=1,\ldots,n$. The \textbf{syndrome vector} \textbf{s} of \textbf{v} is defined by

	\begin{equation}
		\textbf{s}^T = H \cdot \textbf{v}^T = [\textbf{s}_1 \; \; \textbf{s}_2 \; \; \hdots \textbf{s}_m],
	\end{equation} 
	where $\textbf{s}_i\in\mathbb{F}_2^b$, for $i=1,\ldots,m$, and $H$ is the parity-check matrix given in (\ref{equação_matrizH}). Thus:
	\begin{equation} \label{sindrome}
		\textbf{s}_i^T= \sum_{j=1}^{k}A_{ij}\textbf{v}_{j}^T + \textbf{v}_{k+i}^T= \sum_{j=1}^{k}A_{ij}\textbf{e}_j^T + \textbf{e}_{k+i}^T,
	\end{equation}
	for each $i=1,\ldots,m$.
	
	Next, we present Algorithm~\ref{algoritmo_1_erro}, which corrects a
	single symbol error when $m \geq 2$.
	
	
	\begin{algorithm}\caption{Algorithm to correct one symbol error} \label{algoritmo_1_erro}    \textbf{Input:}\\ Positive integers $m$ and $k$, with $m \geq 2$; and a received vector $\textbf{v}=[\textbf{v}_1,\ldots,\textbf{v}_n]$, where  $n=m+k$.
		
		\textbf{Output:}\\ The corrected codeword $\textbf{c}$, or a declaration
		that more than one symbol error has occurred.

		\begin{enumerate}
			\item  \label{Passo1_1} Compute  $\textbf{s}=[\textbf{s}_1 \; \; \textbf{s}_2 \; \; \ldots \textbf{s}_m]$ according to  Equation (\ref{sindrome}). 
			
			\item \label{Passo2_1} If there exists $j \in \{1,\ldots,m\}$ such that $\textbf{s}_j \neq \textbf{0}$ and $\textbf{s}_i = \textbf{0}$ for all $i \neq j$, then set $\textbf{e}_{k+j} = \textbf{s}_j$ and conclude that the error occurred at position $k+j$.
			Otherwise, set $l_1 = 0$.
			
			\item \label{Passo3_1} Set $l_1 = l_1 + 1$. If $l_1 > k$, then terminate the algorithm and declare that more than one symbol error has occurred (decoding failure). Otherwise, proceed to the next step.

			\item \label{Passo4_1} Compute the vectors:
			
			\begin{equation}\label{calculoy}
				\begin{aligned}
					\textbf{y}_1^T &= \textbf{s}_1^T + A_{1{l_1}}A_{m{l_1}}^{-1}\textbf{s}_m^T         \\
					\textbf{y}_2^T &= \textbf{s}_2^T + A_{2{l_1}}A_{1{l_1}}^{-1}\textbf{s}_1^T         \\
					\textbf{y}_3^T &= \textbf{s}_3^T + A_{3{l_1}}A_{2{l_1}}^{-1}\textbf{s}_2^T         \\
					\vdots \\
					\textbf{y}_m^T &= \textbf{s}_m^T + A_{m{l_1}}A_{(m-1){l_1}}^{-1}\textbf{s}_{m-1}^T \\
				\end{aligned}.
			\end{equation}
			
			\item \label{Passo5_1} If $\textbf{y}_1 = \textbf{y}_2 = \cdots = \textbf{y}_m = \textbf{0}$, then there is a symbol error at position $l_1$, with magnitude given by
			\begin{equation} \label{magnitude_erros}
				\textbf{e}_{l_1}^T = A_{i l_1}^{-1}\textbf{s}_i^T,
			\end{equation}
			for any $i=1,\ldots,m$. Otherwise, return to Step~\ref{Passo3_1}.
			
			\item Correct the received vector $\textbf{v}$ by setting $\textbf{c} = \textbf{v} - \textbf{e}$.   
		\end{enumerate}
	\end{algorithm}

	Note that if the syndrome vector $\textbf{s}$ is zero, then no error has occurred. In this case, the algorithm returns $\textbf{c} = \textbf{v}$.

	\begin{example}\label{exemplo1erro}
		Let $p(x)= x^3 + x^2 + 1\in\mathbb{F}_2[x]$ be a primitive polynomial with a primitive element $\alpha \in\mathbb{F}_{2^3}$. The superregular Vandermonde matrix and the Frobenius companion matrix are, respectively, given by 
		\[
		A = \left[ 
		\begin{array}{c c}
			1  & 1  \\
			\alpha  & \alpha^2   
		\end{array}
		\right ] \ \ \mbox{and} \ \ \ 	C= 			\left[ 
		\begin{array}{c c c}
			0  & 0 & 1  \\
			1  & 0 & 0  \\
			0  & 1 & 1  
		\end{array}
		\right ].
		\]
		Thus, according to (\ref{equação_matrizH}),  we can construct an MDS array code $\mathcal{C}$ with parameters $[4,2,3]$ over $\mathbb{F}_2^3$ with parity-check matrix:
		\[
		H =  \left[ 
		\begin{array}{c c | c}
			I_3 & I_3  & \vspace{-0.1cm} \\ \vspace{-0.1cm}
			& & I_6 \\
			C & C^2  & 
		\end{array}
		\right] \]
		
		\noindent This code can correct a single symbol error of length $b=3$, since $\lfloor\frac{m}{2}\rfloor = \lfloor\frac{2}{2}\rfloor = 1$.   
		Assume the received vector $ \textbf{v} = [110 \;\; 110 \;\; 011 \;\; 011]$. According to Step 1 of the Algorithm \ref{algoritmo_1_erro}, we calculate the syndromes using the Equation (\ref{sindrome}):
		\[
		\textbf{s}_1=[0 \ 1 \ 1] \ \ \mbox{and} \ \ \textbf{s}_2=[1 \ 0 \ 0].
		\]
		Since the syndromes obtained are all nonzero, errors may have occurred in the information symbols. So, we make $l_1=1$ and calculate the vectors as in Equation (\ref{calculoy}):  
		\[
		\textbf{y}_1=[0 \ 0 \ 0] \ \ \mbox{and} \ \ \textbf{y}_2=[0 \ 0 \ 0].
		\]
		As both vectors are zero, the error in position $l_1=1$ is confirmed. To obtain the magnitude of this error, calculate the Equation (\ref{magnitude_erros}):
		
		\[
		\textbf{e}_{1}^T= A_{11}^{-1} \textbf{s}_1^T = I_3^{-1} \left [
		\begin{array}{c}
			0 \\
			1 \\
			1 \\
		\end{array}
		\right ] = \left [
		\begin{array}{c c c}
			1 & 0 & 0 \\
			0 & 1 & 0 \\
			0 & 0 & 1 \\
		\end{array}
		\right ] \left [
		\begin{array}{c}
			0 \\
			1 \\
			1 \\
		\end{array}
		\right ] = \left [
		\begin{array}{c}
			0 \\
			1 \\
			1 \\
		\end{array}
		\right ] .
		\]
		Therefore,
		\[
		\textbf{c} = [110 \;\; 110 \;\; 011 \;\; 011] 
		- [\textbf{011} \;\; 000 \;\; 000 \;\; 000] = [101 \;\; 110 \;\; 011 \;\; 011]
		\]
		is the codeword transmitted.
	\end{example}

	Now, we present the algorithm for correcting up to two symbol errors, which requires $m \geq 4$. This algorithm incorporates specific steps of the Algorithm \ref{algoritmo_1_erro}.
	
	\begin{algorithm}\caption{Algorithm to correct two symbol errors} \label{algoritmo_dois_erros}    \textbf{Input:}\\ Positive integers $m$ and $k$, with $m \geq 4$; and a received vector $\textbf{v}=[\textbf{v}_1,\ldots,\textbf{v}_n]$, where  $n=m+k$.
		
		\textbf{Output:}\\ 
		The corrected codeword $c$, or a declaration that more than two symbol errors have occurred, or that decoding has failed.

		\begin{enumerate}
			
			\item \label{Passo1} Compute $\textbf{s}=[\textbf{s}_1 \; \; \textbf{s}_2 \; \; \ldots \textbf{s}_m]$ according to Equation (\ref{sindrome}).
			
			\item \label{Passo2} If there exists $j \in \{1,\ldots,m\}$ such that $\textbf{s}_j \neq \textbf{0}$ and $\textbf{s}_i = \textbf{0}$ for all $i \neq j$, then set $\textbf{e}_{k+j} = \textbf{s}_j$ and conclude that the error occurred at position $k+j$. Otherwise, proceed to the next step.
			
			\item \label{Passo3} If there exist $j_1, j_2 \in \{1,\ldots,m\}$ with $j_1 \neq j_2$ such that $\textbf{s}_{j_1} \neq \textbf{0}$, $\textbf{s}_{j_2} \neq \textbf{0}$, and $\textbf{s}_i = \textbf{0}$ for all $i \neq j_1, j_2$, then set $\textbf{e}_{k+j_1} = \textbf{s}_{j_1}$ and $\textbf{e}_{k+j_2} = \textbf{s}_{j_2}$, and conclude that two symbol errors occurred at positions $k+j_1$ and $k+j_2$. Otherwise, set $l_1 = 0$.
			
		\end{enumerate}		
	\end{algorithm}
	\begin{algorithm}
		\begin{enumerate}
			
			\item[4.] \label{Passo4} Set $l_1 = l_1 + 1$. If $l_1 > k$, then terminate the algorithm and declare that more than two symbol errors have occurred (decoding failure). Otherwise, proceed to the next step.

			\item[5.] \label{Passo5} Compute the vectors of Equation (\ref{calculoy}).

			\item[6.] \label{Passo6} If $\textbf{y}_1 = \textbf{y}_2 = \cdots = \textbf{y}_m = \textbf{0}$, then there is a symbol error at position $l_1$, with magnitude given by
			\begin{equation}\label{erro_posicao_l1}
				\textbf{e}_{l_1}^T = A_{i l_1}^{-1} \textbf{s}_i^T,
			\end{equation}
			for any $i=1,\ldots,m$. Otherwise, proceed to the next step.
			
			\item[7.] \label{Passo7} If $(\textbf{y}_i,\textbf{y}_{i+1})=(\textbf{0},\textbf{0})$ for some $i=1,\ldots,m-1$ or $(\textbf{y}_m,\textbf{y}_1)=(\textbf{0},\textbf{0})$, but not all $\textbf{y}_i$ are zero, then one symbol error occurs at position $l_1$ and a second symbol error occurs in the parity symbols, given by
			\begin{equation}\label{Passo7:erro_informação}
				\textbf{e}_{l_1}^T = A_{i l_1}^{-1} \textbf{s}_i^T
			\end{equation}
			and
			\begin{equation}\label{Passo7:erro_paridade}
				\textbf{e}_{k+j} = \textbf{s}_j - A_{j l_1} \textbf{e}_{l_1},
			\end{equation}
			for any $i \neq j$, where $j$ denotes the index corresponding to the nonzero syndrome component associated with the parity error. Otherwise, proceed to the next step.

			\item[8.] \label{Passo8} For each $l_2 \in \{l_1 + 1, \ldots, k\}$, check whether
			\begin{equation} \label{eq_y} 
				\textbf{y}_i^T = C^{r_{i-2}}\textbf{y}_{i-1}^T,
			\end{equation} 
			for all $i= 3, 4, \ldots, m$ with 
			\begin{equation} \label{eq_r}
				\begin{aligned}
					r_{i-2} = & \hspace{.15cm} \sigma(i,l_2) - \sigma(i-1,l_2) + Z(\sigma(i,l_1) -\sigma(i,l_2) - \sigma(i-1,l_1) + \sigma(i-1,l_2) ) \\  
					&  - Z(\sigma(i-1,l_1) - \sigma(i-1,l_2) -\sigma(i-2,l_1)   + \sigma(i-2,l_2)).            \end{aligned}
			\end{equation}
			If this condition holds, then two symbol errors are declared at positions $l_1$ and $l_2$. Otherwise, return to Step~\ref{Passo4}.
			
			\item[9.] \label{Passo9} To determine the magnitudes of the errors, solve
			the linear system.
			\begin{equation} \label{sistema}
				\left \{ \begin{matrix} A_{il_1}\textbf{e}_{l_1}^T + A_{il_2}\textbf{e}_{l_2}^T = \textbf{s}_i^T  \\ A_{jl_1}\textbf{e}_{l_1}^T + A_{jl_2}\textbf{e}_{l_2}^T = \textbf{s}_j^T  \end{matrix} \right.,        
			\end{equation}
			for any $i,j = 1, 2,\ldots, m$ with $i \neq j$.

			\item[10.] Correct the received vector $\textbf{v}$ by setting
			$\textbf{c} = \textbf{v} - \textbf{e}.$ 
		\end{enumerate}		
	\end{algorithm}

	The following theorem demonstrates that Algorithm~\ref{algoritmo_dois_erros} indeed corrects up to two symbol errors of length $b$.

	\begin{theorem}\label{casos2erros}
		Let $\mathcal{C}$ be an MDS array code over $\mathbb{F}_2^b$ with
		parameters $[m+k,k,m+1]$, where $m \ge 4$ and $k,b$ are positive integers, and
		let $H$ be its parity-check matrix given in
		(\ref{equação_matrizH}). Then Algorithm~\ref{algoritmo_dois_erros}
		corrects up to two symbol errors with unknown locations.
	\end{theorem}
	
	\begin{proof}
		Since $n = m + k$, symbol errors in the received vector
		$\textbf{v} = [\textbf{v}_1, \ldots, \textbf{v}_n]$ may affect either the
		$k$ information symbols or the $m$ parity symbols. All such cases are
		considered below.

		
		\begin{itemize}
			\item \label{caso1} \textit{ \textbf{Case 1: One or two errors in parity symbols}}. If one or two errors occur in parity symbols, then $m-1$ or $m-2$
			syndromes are zero, respectively. If the errors occur at positions
			$k+j_1$ and $k+j_2$, then the nonzero syndromes are $\textbf{s}_{j_1}$
			and $\textbf{s}_{j_2}$, respectively.
			
			So, the magnitudes of the  errors are given by
			\begin{equation}
				\textbf{e}_{k+j_1} = \textbf{s}_{j_1} \;\; \mbox{and}  \;\; \textbf{e}_{k+j_2} = \textbf{s}_{j_2}.
			\end{equation}

			
			\item \label{caso2}  \textit{\textbf{Case 2: One error in an information symbol.}} Suppose that a single error occurred in an information symbol at position $l_1$. Then, by (\ref{sindrome}), the syndromes are given by 
			\begin{equation}   \textbf{s}_i^T=A_{il_1}\textbf{e}_{l_1}^T, 
			\end{equation}
			for all $i=1, \ldots , m$. Replacing in (\ref{calculoy}), we have $   y_i^T= 0$, for all $i=1,..,m$, because            \begin{equation} \label{erro}
				\textbf{e}_{l_1}^T= A_{1l_1}^{-1}\textbf{s}_1^T= A_{2l_1}^{-1} = \ldots = A_{ml_1}^{-1}\textbf{s}_m^T,
			\end{equation}
			and the magnitude of the  error is obtained using any equality from Equation (\ref{erro}).

			
			\item  \label{caso3} \textit{\textbf{Case 3: One error in an information symbol and one error in a parity symbol.}}
			Assume that the error in the information symbol occurred at position $l_1$ and, without loss of generality, that the error in the parity symbol occurred at position $k+1$.  Then, by (\ref{sindrome}), the syndromes are given by
			\begin{equation}
				\textbf{s}_1^T= A_{1l_1}\textbf{e}_{l_1}^T + \textbf{e}_{k+1} \ \ \ \  \mbox{and}  \ \ \ \  \textbf{s}_i^T=A_{il_1}\textbf{e}_{l_1}^T, \;  
			\end{equation}
			for all $i=2,\ldots,m$. In this case, we will have $\textbf{y}_1,\textbf{y}_2 \neq\textbf{0}$, since $\textbf{s}_1$ is involved in the computation of these vectors, and $\textbf{y}_3=\ldots=\textbf{y}_m=0$.
			The magnitude of the error in the information symbol is obtained from any syndrome equation with $i=2,\ldots,m$, namely,
			\begin{equation}
				\textbf{e}_{l_1}^T = A_{il_1}^{-1}\textbf{s}_i^T.
			\end{equation}
			Substituting this value into the first syndrome equation, the magnitude of the error in the parity symbol is given by
			\begin{equation}
				\textbf{e}_{k+1} = \textbf{s}_1^T - A_{1l_1}\textbf{e}_{l_1}^T.
			\end{equation}

			\item \label{caso4} \textit{\textbf{Case 4: Two errors in information symbols.}} Suppose now that the two errors in the information symbols are at positions $l_1$ and $l_2$. By (\ref{sindrome}), the syndromes are calculated by
			\begin{equation}\label{sindrome_2e}
				\textbf{s}_i^T=A_{il_1}\textbf{e}_{l_1}^T + A_{il_2}\textbf{e}_{l_2}^T,
			\end{equation}
			for $i=1,...,m$. Without loss of generality, suppose $l_1=1$. Replacing (\ref{sindrome_2e}) into the vectors $\textbf{y}_i$ given by (\ref{calculoy}), then $\textbf{y}_i^T \ne 0$, for all $i=1,\ldots,m$. Fixing $\textbf{y}_2$ we have 
			{\small       \begin{equation}
					\begin{aligned}
						\textbf{y}_2^T &= \textbf{s}_2^T + A_{2l_1}A_{1l_1}^{-1}\textbf{s}_1^T \\
						&= (A_{2l_1}\textbf{e}_{l_1}^T + A_{2l_2}\textbf{e}_{l_2}^T) + A_{2l_1}A_{1l_1}^{-1}(A_{1l_1}\textbf{e}_{l_1}^T + A_{1l_2}\textbf{e}_{l_2}^T) \\
						&= A_{2l_1}\textbf{e}_{l_1}^T + A_{2l_2}\textbf{e}_{l_2}^T + A_{2l_1}A_{1l_1}^{-1}A_{1l_1}\textbf{e}_{l_1}^T + A_{2l_1}A_{1l_1}^{-1}A_{1l_2}\textbf{e}_{l_2}^T \\
						&= A_{2l_2}\textbf{e}_{l_2}^T + A_{2l_1}A_{1l_1}^{-1}A_{1l_2}\textbf{e}_{l_2}^T \\
						&= (A_{2l_2}+A_{2l_1}A_{1l_1}^{-1}A_{1l_2})\textbf{e}_{l_2}^T.
					\end{aligned}
			\end{equation}}
			\noindent Since $A_{ij}=C^{\sigma(i,j)}$, then
			\begin{equation}
				\textbf{y}_2^T= (C^{\sigma(2,l_2)}+C^{\sigma(2,l_1)-\sigma(1,l_1)+\sigma(1,l_2)})\textbf{e}_{l_2}^T.
			\end{equation}
			Using the Zech logarithm defined in (\ref{definição_zech}) and the property (\ref{propriedade_zech}), we have
			\begin{equation} \label{y2}
				\textbf{y}_2^T= C^{\sigma(2,l_2)+Z(\sigma(2,l_1)-\sigma(1,l_1)+\sigma(1,l_2)-\sigma(2,l_2))}\textbf{e}_{l_2}^T,
			\end{equation} that is 
			\small \begin{equation} \label{errol2}
				\textbf{e}_{l_2}^T = C^{- \sigma(2,l_2) - Z(\sigma(2,l_1) - \sigma(1,l_1) + \sigma(1,l_2)-\sigma(2,l_2))} \textbf{y}_2^T.
			\end{equation}
			In the same way
			\begin{equation} \label{y3}
				\textbf{y}_3^T= C^{\sigma(3,l_2)+Z(\sigma(3,l_1)-\sigma(2,l_1)+\sigma(2,l_2)-\sigma(3,l_2))}\textbf{e}_{l_2}^T.
			\end{equation}
			Replacing (\ref{errol2}) into (\ref{y3}), we have
			\begin{equation}\label{demy3}
				\textbf{y}_3^T = C^{r_1}\textbf{y}_2^T,
			\end{equation}
			where 
			\begin{equation}\label{demr1}
				\begin{aligned}
					r_1 &= \sigma(3,l_2) - \sigma(2,l_2) + Z(\sigma(3,l_1) - \sigma(3,l_2) - \sigma(2,l_1) + \sigma(2,l_2)) \\ 
					&- Z(\sigma(2,l_1) - \sigma(2,l_2) - \sigma(1,l_1) + \sigma(1,l_2)).
				\end{aligned}
			\end{equation}
			
			Similarly,
			\begin{equation}\label{demyi}
				\textbf{y}_i^T=C^{r_{i-2}}\textbf{y}_{i-1}^T,
			\end{equation}
			for all $i=3,\ldots,m$, where
			\begin{equation}\label{demri}
				\begin{aligned}
					r_{i-2} &= \sigma(i,l_2) - \sigma(i-1,l_2) + Z(\sigma(i,l_1) - \sigma(i,l_2) 
					- \sigma(i-1,l_1)   \\ 
					& + \sigma(i-1,l_2)) - Z(\sigma(i-1,l_1)  - \sigma(i-1,l_2) - \sigma(i-2,l_1) \\
					& + \sigma(i-2,l_2)).
				\end{aligned}
			\end{equation}
			By ensuring that Equation (\ref{demyi}) hold for all $i=3,\ldots,m$, it
			follows that the expressions $\textbf{s}_i$ are valid, that is, the
			symbol errors indeed occurred at positions $l_1$ and $l_2$. To determine
			the magnitudes of the errors, solve the linear system given in
			Equation (\ref{sistema}) using any two syndrome equations. 
			To determine the magnitudes of the errors, solve the linear system given in Equation~(\ref{sistema}) using any two syndrome equations.  Its coefficient matrix is a $2b\times 2b$ block submatrix of $\psi(A)$ and is nonsingular due to the block-superregularity induced by the superregularity of $A$. Consequently, the error magnitudes $e_{l_1}$ and $e_{l_2}$ are uniquely determined.
		\end{itemize}
	\end{proof}

	Below, we present examples illustrating Cases 3 and 4 of Algorithm \ref{algoritmo_dois_erros}, in order to better elucidate the behavior of the  algorithm using MDS array codes constructed from superregular Vandermonde and Cauchy matrices.
	

	\begin{example} \label{exemplo_caso_3}
		Consider $\mathcal{C}$ the MDS array code with parameters  $[10,5,6]$ over $\mathbb{F}_2^5$ constructed in Example \ref{Exemplo: construção H}. Assume that the received vector is $$\textbf{v}= [01101 \;\; 11101 \;\; 10110 \;\; 11110 \;\; 10101 \;\; 01011 \;\; 01000 \;\; 10100 \;\; 01111 \;\; 10011].$$ The syndromes from (\ref{sindrome})  are given by 
		\[
		\textbf{s}^T_1 = \left [
		\begin{array}{c}
			0 \\
			0 \\
			1 \\
			1 \\
			0 \\
		\end{array}
		\right ], \ \
		\textbf{s}^T_2 = \left [
		\begin{array}{c}
			0 \\
			1 \\
			1 \\
			0 \\
			0 \\
		\end{array}
		\right ], \ \ 
		\textbf{s}^T_3 = \left [
		\begin{array}{c}
			0 \\
			0 \\
			1 \\
			1 \\
			0 \\
		\end{array}
		\right ], \ \
		\textbf{s}^T_4 = \left [
		\begin{array}{c}
			0 \\
			0 \\
			0 \\
			1 \\
			1 \\
		\end{array}
		\right ], \ \
		\textbf{s}^T_5 = \left [
		\begin{array}{c}
			1 \\
			1 \\
			1 \\
			1 \\
			1 \\
		\end{array}
		\right ].
		\]
		Since all the syndromes are nonzero, we start by calculating the equations in (\ref{calculoy}) with $l_1=1$, yielding the following results:
		\[
		\textbf{y}^T_1 = \left [
		\begin{array}{c}
			1 \\
			1 \\
			1 \\
			1 \\
			0 \\
		\end{array}
		\right ], \ \
		\textbf{y}^T_2 = \left [
		\begin{array}{c}
			0 \\
			1 \\
			1 \\
			1 \\
			1 \\
		\end{array}
		\right ], \ \ 
		\textbf{y}^T_3 = \left [
		\begin{array}{c}
			0 \\
			0 \\
			0 \\
			0 \\
			0 \\
		\end{array}
		\right ], \ \
		\textbf{y}^T_4 = \left [
		\begin{array}{c}
			0 \\
			0 \\
			0 \\
			0 \\
			0 \\
		\end{array}
		\right ], \ \
		\textbf{y}^T_5 = \left [
		\begin{array}{c}
			0 \\
			0 \\
			0 \\
			0 \\
			0 \\
		\end{array}
		\right ].
		\]
		Note that in this case, we have three zero vectors, confirming an error at position $l_1 = 1$. The two nonzero vectors share the syndrome $\textbf{s}_1$. Therefore, the second error occurred in the position corresponding to the parity symbols. This example falls under Case 3. To correct it, we find the error at position $l_1 = 1$ using  (\ref{Passo7:erro_informação}) and the error in the position of the parity symbols using (\ref{Passo7:erro_paridade}). They are respectively given by
		\[
		\textbf{e}_{1}^T=A_{21}^{-1}\textbf{s}_2=
		\left [
		\begin{array}{c}
			1  \\
			1  \\
			0  \\
			0  \\
			0  \\
		\end{array}
		\right ] \ \ \mbox{and} \ \ \
		\textbf{e}_{6}^T=\textbf{s}_1 - A_{11}\textbf{e}_1 =
		\left [
		\begin{array}{c}
			1  \\
			1  \\
			1  \\
			1  \\
			0  \\
		\end{array}
		\right ].
		\]
		Therefore, the codeword is given by
		\[
		\begin{aligned}
			\textbf{c} = \textbf{v} - \textbf{e} & = [01101 \;\; 11101 \;\; 10110 \;\; 11110 \;\; 10101 \;\; 01011 \;\; 01000 \;\; 10100 \;\; 01111 \;\; 10011] \\
			&- [\textbf{11000} \;\; 00000 \;\; 00000 \;\; 00000 \;\; 00000 \;\; \textbf{11110} \;\; 00000 \;\; 00000 \;\; 00000 \;\; 00000] \\
			& = [10101 \;\; 11101 \;\; 10110 \;\; 11110 \;\; 10101 \;\; 10101 \;\; 01000 \;\; 10100 \;\; 01111 \;\; 10011].
		\end{aligned}
		\]           
		
	\end{example}

	\begin{example} \label{Exemplo: Cauchy 1}
		Let $p(x) = x^4 + x^3 + 1 \in \mathbb{F}_2[x]$ be a primitive polynomial. Through this polynomial it is possible to construct an MDS array code $\mathcal{C}$ with parameter $[8,4,5]$ over $\mathbb{F}_2^4$ capable for correcting up to two symbol errors of length $b=4$ by using a superregular matrix. For this example, we will use a superregular Cauchy  matrix $A$ and the Frobenius companion matrix $C$ to associated with $p(x)$, given, respectively, by
		\[
		A = \left [
		\begin{array}{c c c c}
			\alpha^{12} & \alpha^{11} & \alpha^{10} & \alpha^{9} \\
			\alpha & \alpha^{14} & \alpha^{10} & \alpha^{13} \\
			\alpha^{5} & 1 & \alpha^{13} & \alpha^{4}  \\
			\alpha^{9} & \alpha^{5} & \alpha^{14} & \alpha^{12}  \\
		\end{array}
		\right ] \ \ \mbox{and} \ \ \	C =
		\left [
		\begin{array}{c c c c }
			0 & 0 & 0  & 1 \\
			1 & 0 & 0  & 0 \\
			0 & 1 & 0  & 0 \\
			0 & 0 & 1  & 1 \\
		\end{array}
		\right ],
		\]
		where each entry of the matrix $A$ is computed according to
		Definition~\ref{defcauchy}, that is, as $(x_i + y_j)^{-1}$ in
		$\mathbb{F}_{2^4}$ and expressed in terms of powers of the primitive element $\alpha$. For instance, since $x_1 + y_1 = \alpha^3$, we have $a_{11} = (\alpha^3)^{-1} = \alpha^{12}$. This representation is used for convenience in subsequent computations.
		So, according to (\ref{equação_matrizH}), the parity-check matrix is given by
		\[H = \left [
		\begin{array}{c c c c | c}
			C^{12} & C^{11} & C^{10} & C^9    &          \\
			C^{11} & C^{12} & C^5    & C^7    &  I_{16} \\
			C^5    & C^{10} & C^{11} & C^{4}  &          \\
			C      & C^4    & C^{9}  & C^{10} & 
		\end{array}
		\right ]. \] 
		Now, suppose we receive the vector $$\textbf{v}= [1011  \;\; 0101 \;\; 0111 \;\; 1001 \;\; 1000 \;\; 1011 \;\; 0111 \;\; 1001].$$ The syndromes from (\ref{sindrome}) are given by
		\[
		\textbf{s}^T_1 = \left [
		\begin{array}{c}
			0 \\
			1 \\
			1 \\
			1 \\
		\end{array}
		\right ],
		\textbf{s}^T_2 = \left [
		\begin{array}{c}
			0 \\
			1 \\
			1 \\
			0 \\
		\end{array}
		\right ],
		\textbf{s}^T_3 = \left [
		\begin{array}{c}
			0 \\
			0 \\
			0 \\
			0 \\
		\end{array}
		\right ],
		\textbf{s}^T_4 = \left [
		\begin{array}{c}
			0 \\
			1 \\
			1 \\
			1 \\
		\end{array}
		\right ].
		\]	
		As only one of the syndromes is zero, considering $l_1=1$, we follow the algorithm by calculating the vectors expressed in Equation (\ref{calculoy}),
		\[
		\textbf{y}_1^T = \left [
		\begin{array}{c}
			1 \\
			1 \\
			1 \\
			0 \\
		\end{array}
		\right ],
		\textbf{y}_2^T= \left [
		\begin{array}{c}
			1 \\
			0 \\
			0 \\
			0 \\
		\end{array}
		\right ], 
		\textbf{y}_3^T = \left [
		\begin{array}{c}
			1 \\
			1 \\
			1 \\
			0 \\
		\end{array}
		\right ],
		\textbf{y}_4^T = \left [
		\begin{array}{c}
			0 \\
			1 \\
			1 \\
			1 \\
		\end{array}
		\right ].
		\]
		Since all of them are nonzero, errors occurred in the information symbols. So, this example falls on Case 4. The algorithm is followed by assuming a second error in $l_2=2$.  So we have to test the Equation (\ref{eq_y}) through the calculations of $r_{i-2}$ as in (\ref{eq_r}), for $i=3,4$: 
		\[
		\begin{aligned}
			r_1 &= \sigma(3, 2) - \sigma(2, 2)                                   
			+ Z(\sigma(3, 1) - \sigma(3, 2) - \sigma(2, 1) + \sigma(2, 2))  
			- Z(\sigma(2, 1) - \sigma(2, 2)  \\ & - \sigma(1, 1)+ \sigma(1, 2))  = 10 - 12 +Z(5 - 10 - 11 + 12) -Z( 11 - 12 - 12 + 11)
			\\ &= -2 + Z(-4) - Z(-2) = 5,
		\end{aligned}
		\]
		and
		\[
		\begin{aligned}
			r_2 &= \sigma(4, 2) - \sigma(3, 2) 
			+ Z(\sigma(4, 1) - \sigma(4, 2) - \sigma(3, 1) + \sigma(3, 2)) 
			- Z(\sigma(3, 1) - \sigma(3, 2) \\ & - \sigma(2, 1)+ \sigma(2, 2)) = 4 - 10 +Z(1 - 4 - 5 + 10) - Z(5 - 10 - 11 + 12)
			\\ & = -6 + Z(2) - Z(-4) = -6 +9 -14 = -11.
		\end{aligned}
		\]
		Then   
		\[\textbf{y}_3^T \neq C^5  \textbf{y}_2^T \ \ \ \mbox{and} \ \ \  \textbf{y}_4^T \neq C^{-11}  \textbf{y}_3^T.\]
		\noindent As for $l_2=2$ the above equations fail we test to $l_2=3$,
		\[   \begin{aligned}
			r_1 &= \sigma(3, 3) - \sigma(2, 3)                                   
			+ Z(\sigma(3, 1) - \sigma(3, 3) - \sigma(2, 1) + \sigma(2,3))  
			- Z(\sigma(2, 1) - \sigma(2, 3)  \\ & - \sigma(1, 1)+ \sigma(1, 3))  = 11 - 5 +Z(5 - 11 - 11 + 5) -Z( 11 - 5 - 12 + 10)
			\\ &= 6 + Z(-12) - Z(4) = 7
		\end{aligned}
		\]	
		and
		\[
		\begin{aligned}
			r_2 &= \sigma(4, 3) - \sigma(3, 3) 
			+ Z(\sigma(4, 1) - \sigma(4, 3) - \sigma(3, 1) + \sigma(3, 3)) 
			- Z(\sigma(3, 1) - \sigma(3, 3) \\ & - \sigma(2, 1)+ \sigma(2, 3)) = 9 - 11 +Z(1 - 9 - 5 + 11) - Z(5 - 11 - 11 + 5)
			\\ & = -2 + Z(-2) - Z(-12) = -2 +7 -4 = 1.
		\end{aligned}
		\]
		Then we confirm that
		\[ \textbf{y}_3^T  = C^7 \textbf{y}_2^T \ \ \ \mbox{and} \ \ \ \textbf{y}_4^T  = C^{1}   \textbf{y}_3^T.\]
		\noindent Therefore, there are errors in positions $l_1=1$ and $l_2=3$.  In order to obtain the magnitude of the errors we solve the linear system from (\ref{sistema}) and we obtain the errors
		\[\textbf{e}_{1}= [01  1  0] \ \ \ \mbox{and} \ \ \ \textbf{e}_{3}= [1  1  1  0].\]
		And so, the codeword is
		\[
		\begin{aligned}
			\textbf{c} & = [1011  \;\; 0101 \;\; 0111 \;\; 1001 \;\; 1000 \;\; 1011 \;\; 0111 \;\; 1001] \\
			&- [\textbf{0110} \;\; 0000 \;\; \textbf{1110} \;\; 0000 \;\; 0000 \;\; 0000 \;\; 0000 \;\; 0000] \\
			& = [1101  \;\; 0101 \;\; 1001 \;\; 1001 \;\; 1000 \;\; 1011 \;\; 0111 \;\; 1001].
		\end{aligned}
		\]
	\end{example}
	
	One of the main computational challenges of Algorithm~\ref{algoritmo_dois_erros}, as analyzed in Theorem~\ref{casos2erros}, arises in Case~4, where it is necessary to compute the quantities in (\ref{demri}) in order to substitute them into (\ref{demyi}). The following corollary provides a direct expression for these terms when $A$ is a superregular Vandermonde matrix.

	\begin{corollary}\label{corollary}
		If $A$ is superregular a Vandermonde matrix and $H = \big[\psi(A) \ | \  I_{(n-k)b}\big]$ follows the construction in Theorem~\ref{teoSara}, then Equation (\ref{demri}) simplifies to
		\begin{equation}
			\label{eq:l2}
			r_{i-2}=l_2,
		\end{equation}
		for all $i= 3, 4, \ldots, m$, where $l_2$ is the position of the second error.
	\end{corollary}
	
	\begin{proof}
		Let $A$ be a Vandermonde matrix satisfying the superregularity conditions of Theorem~\ref{teoSara}. We will initially show that the calculation of Zech logarithms, from Equation (\ref{demri}), cancel out for all $i=3,\ldots,m$, that is,
		\begin{equation}\label{eq:z}
			\begin{aligned}
				& Z(\sigma (i,l_1)-\sigma (i,l_2)-\sigma (i-1,l_1)+\sigma (i-1,l_2)) = \\
				& Z(\sigma (i-1,l_1)-\sigma (i-1,l_2)-\sigma (i-2,l_1)+\sigma (i-2,l_2)).
			\end{aligned}    
		\end{equation}
		This is equivalent to showing that the arguments are equal,
		which implies that \begin{equation}\label{argumentozech}
			\sigma (i,l_1)-2\sigma (i-1,l_1)+\sigma (i-2,l_1)-\sigma (i,l_2)+2\sigma (i-1,l_2) - \sigma (i-2,l_2)=0.
		\end{equation}
		\noindent Since $A$ is a Vandermonde matrix, then by (\ref{vandermonde}),  if $l$ is any column of $A$, considering consecutive rows $i-2$, $i-1$, and $i$, then
		\begin{equation}\label{colunal}
			\sigma (i,l) =  2\sigma (i-1,l)-\sigma (i-2,l).
		\end{equation}   
		Replacing (\ref{colunal})  into (\ref{argumentozech}), we obtain that  the Zech logarithms indeed cancels out. Therefore, Equation (\ref{demri}) can be simplified to
		\begin{equation}
			\label{eq:sigma}
			\begin{split}
				r_{i-2}= \sigma (i,l_2)- \sigma (i-1,l_2),
			\end{split}
		\end{equation}
		for all $i=3,\ldots,m$. Again, due to $A$ being a Vandermonde matrix, if we consider two consecutive rows $i-1$ and $i$, and a fixed column $l_2$, it follows that the difference between the powers of $\alpha$ in $\sigma (i,l_2)$ and $\sigma (i-1,l_2)$ is equal to $l_2$, which provides us (\ref{eq:l2}), concluding the proof.
	\end{proof}
	
	In the following example we show how calculations are made easier by using the Corollary \ref{corollary}. 
	
	\begin{example} \label{ex_caso_4}
		Consider again $\mathcal{C}$ the MDS array code with parameters  $[10,5,6]$ over $\mathbb{F}_2^5$ constructed in Example \ref{Exemplo: construção H}.  Now, suppose that the received vector is \[\textbf{v}= [11001 \;\; 11101 \;\; 11100 \;\; 11110 \;\; 10101 \;\; 01011 \;\; 01000 \;\; 10100 \;\; 01111 \;\; 10011].\] The syndromes from (\ref{sindrome}) are given by 
		\[
		\textbf{s}^T_1 = \left [
		\begin{array}{c}
			0 \\
			0 \\
			1 \\
			1 \\
			0 \\
		\end{array}
		\right ], \ 
		\textbf{s}^T_2 = \left [
		\begin{array}{c}
			0 \\
			1 \\
			0 \\
			0 \\
			0 \\
		\end{array}
		\right ], \  
		\textbf{s}^T_3 = \left [
		\begin{array}{c}
			1 \\
			0 \\
			0 \\
			1 \\
			1 \\
		\end{array}
		\right ], \ 
		\textbf{s}^T_4 = \left [
		\begin{array}{c}
			1 \\
			1 \\
			1 \\
			0 \\
			1 \\
		\end{array}
		\right ], \ 
		\textbf{s}^T_5 = \left [
		\begin{array}{c}
			1 \\
			1 \\
			1 \\
			1 \\
			0 \\
		\end{array}
		\right ].
		\]
		Since all syndromes are nonzero, conducting the calculations in (\ref{calculoy}), initially with $l_1 = 1$, we obtain
		\[
		\textbf{y}^T_1 = \left [
		\begin{array}{c}
			0 \\
			1 \\
			1 \\
			1 \\
			0 \\
		\end{array}
		\right ], \ 
		\textbf{y}^T_2 = \left [
		\begin{array}{c}
			0 \\
			1 \\
			0 \\
			1 \\
			1 \\
		\end{array}
		\right ], \ 
		\textbf{y}^T_3 = \left [
		\begin{array}{c}
			1 \\
			0 \\
			1 \\
			1 \\
			1 \\
		\end{array}
		\right ], \ 
		\textbf{y}^T_4 = \left [
		\begin{array}{c}
			0 \\
			1 \\
			0 \\
			1 \\
			0 \\
		\end{array}
		\right ], \ 
		\textbf{y}^T_5 = \left [
		\begin{array}{c}
			0 \\
			1 \\
			1 \\
			1 \\
			0 \\
		\end{array}
		\right ].
		\]
		By Corollary \ref{corollary} all values of $r_{i-2}$ is equal to $l_2$, that is, $r_1 = r_2 = r_3 = 2$.  Then, 
		\[\textbf{y}_i^T \neq C^2 \textbf{y}_{i-1}^T,\]
		for all $i=3,4,5$.
		As the above equalities fail, the algorithm follows for $l_2=3$, where again by using Corollary \ref{corollary}, we have $r_1 = r_2 = r_3 = 3$. Then, we check that
		\[\textbf{y}_i^T  = C^3 \textbf{y}_{i-1}^T,\]
		for all $i=3,4,5$.
		Therefore, by (\ref{sistema}), we obtain the magnitude of the errors $\textbf{e}_1= [01100]$ and $\textbf{e}_3= [01010]$. Thus, the correct codeword is
		\begin{equation*} \small
			\begin{aligned}
				\textbf{c} = \textbf{v} - \textbf{e} & = 
				[11001 \; 11101 \; 11100 \; 11110 \; 10101 \; 10101 \; 01000 \; 10100 \; 01111 \; 10011] \\
				&-  [\textbf{01100} \; 00000 \; \textbf{01010} \; 00000 \; 00000 \; 00000 \; 00000 \; 00000 \; 00000 \; 00000] \\
				& = [10101 \; 11101 \; 10110 \; 11110 \; 10101 \; 10101 \; 01000 \; 10100 \; 01111 \; 10011].
			\end{aligned}
		\end{equation*}
		
	\end{example}
	
	We note that the focus of this work is on the algebraic construction and decoding capability of MDS array codes, rather than on a detailed complexity analysis. Nevertheless, the proposed decoding procedures involve operations over $\mathbb{F}_q^b$ whose computational cost is dominated by syndrome evaluations and the solution of small linear systems. For moderate values of $m$ and $k$, these operations remain efficient and suitable for practical implementation.
	In particular, when Vandermonde matrices are employed, the decoding expressions simplify significantly, reducing the number of required operations. 


	
	\subsection{Three Symbol Errors} \label{seção_correção_três_erros}

	In this subsection, we analyze the extension of the proposed decoding framework to the case of three symbol errors. The different error configurations include the cases where all three
	errors occur in parity symbols, where one error occurs in an information symbol and two in parity symbols, and where two errors occur in information symbols and one in a parity symbol. These cases can be treated by arguments analogous to those developed for the two-error case in Algorithm~\ref{algoritmo_dois_erros}, and their derivations follow
	similar algebraic steps. For this reason, we focus on the most representative and technically involved scenario, namely when all three errors occur in the information symbols.
	
	Following a procedure similar to the proof of
	Case~4 in Theorem~\ref{casos2erros}, let $l_1$, $l_2$, and $l_3$ denote
	the positions of the three errors in the information symbols. The
	corresponding syndromes are then given by
	\begin{equation}\label{sindromes_3_erros}
		\textbf{s}_i = A_{il_1} \textbf{e}_{l_1}^T + A_{il_2} \textbf{e}_{l_2}^T + A_{il_3} \textbf{e}_{l_3}^T,
	\end{equation}
	for $i=1,...,m$. Without loss of generality, suppose $l_1=1$.
	Replacing (\ref{sindromes_3_erros}) in vectors $\textbf{y}_i$ as given in (\ref{calculoy}), then $\textbf{y}_i^T \ne \textbf{0}$, for all $i=1,\ldots,m$.  For $i=2,...,m$, we have
	{\small \begin{equation}
			\begin{aligned}
				\textbf{y}_i^T &= \textbf{s}_i^T + A_{i{l_1}}A_{(i-1){l_1}}^{-1}\textbf{s}_{i-1}^T \\
				&=\left(A_{il_2}+A_{il_1}A_{(i-1)l_1}^{-1}A_{(i-1)l_2}\right)\textbf{e}_{l_2}^T+\left(A_{il_3}+A_{il_1}A_{(i-1)l_1}^{-1}A_{(i-1)l_3}\right)\textbf{e}_{l_3}^T \\
				&= \left(C^{\sigma(i, l_{2})} + C^{\sigma(i, l_{1}) - \sigma(i-1, l_{1}) + \sigma(i-1, l_{2})}\right)\textbf{e}_{l_{2}}^{T} + \left(C^{\sigma(i, l_{3})} + C^{\sigma(i, l_{1}) - \sigma(i-1, l_{1}) + \sigma(i-1, l_{3})}\right)\textbf{e}_{l_{3}}^{T} \\
				&= C^{r_{i1}}\textbf{e}_{l_2}^T+ C^{r_{i2}}\textbf{e}_{l_3}^T,
			\end{aligned}
	\end{equation}}
	
	\noindent  where
	\begin{equation}\label{ri1}
		r_{i1}=\sigma(i,l_2)+Z\left(\sigma(i,l_1)-\sigma(i,l_2)-\sigma(i-1,l_1)+\sigma(i-1,l_2)\right)
	\end{equation}
	and
	\begin{equation}\label{ri2}
		r_{i2}=\sigma\left(i,l_3)+Z(\sigma(i,l_1)-\sigma(i,l_3)-\sigma(i-1,l_1)+\sigma(i-1,l_3)\right).
	\end{equation}
	Let us consider three consecutive indices $i-2$, $i-1$, and $i$, for $i=4,...,m$. Isolating $\textbf{e}_{l_{2}}^{T}$ in $\textbf{y}_{i-2}^T$, we have:
	\begin{equation}
		\label{el2}
		\textbf{e}_{l_{2}}^{T} = \left(\textbf{y}_{i-2}^{T} -  C^{r_{(i-2)2}} \textbf{e}_{l_{3}}^{T}\right)C^{-r_{(i-2)1}}.
	\end{equation}
	Substituting (\ref{el2}) into $\textbf{y}_{i-1}^T$ and isolating $\textbf{e}_{l_{3}}^{T}$, then
	\begin{equation}\label{el3}
		\mathbf e_{l_3}^{T}
		=
		\left(\mathbf y_{i-1}^{T} - C^{\overline{r_{(i-1)1}}}\mathbf y_{i-2}^{T}\right)C^{-\overline{r_{(i-1)2}}}.
	\end{equation}
	where
	\begin{equation}\label{bari1}
		\overline{r_{(i-1)1}} =  r_{(i-1)1} - r_{(i-2)1}
	\end{equation}
	and
	\begin{equation}\label{bari2}
		\overline{r_{(i-1)2}} = r_{(i-1)2} + Z(r_{(i-1)1}+r_{(i-2)2}-r_{(i-2)1}-r_{(i-1)2}).
	\end{equation}
	Finally, substituting (\ref{el2}) and (\ref{el3}) into $\textbf{y}_{i-2}^T$, we have:
	\begin{equation}\label{relacaoy3erros}    \textbf{y}_i^T=C^{\widehat{r_{i1}}}\textbf{y}_{i-1}^T+C^{\widehat{r_{i2}}}\textbf{y}_{i-2}^T,
	\end{equation}
	where
	\begin{equation}\label{ri1hat}
		\widehat{r_{i1}} =
		r_{i2}
		- \overline{r_{(i-1)2}}
		+
		Z\big(
		r_{i1}+r_{(i-2)2}-r_{(i-2)1}-r_{i2}
		\big)
	\end{equation}
	and
	\begin{equation}\label{ri2hat}
		\begin{aligned}
			\widehat{r_{i2}} &= r_{i1}-r_{(i-2)1}+Z(-r_{(i-2)2}-\overline{r_{(i-1)1}}-\overline{r_{(i-1)2}}\\
			&\quad +Z(r_{i2}+Z(r_{i1}+r_{(i-2)2}-r_{(i-2)1}-r_{i2})+\overline{r_{(i-1)1}}-\overline{r_{(i-1)2}}-r_{i1}+r_{(i-2)1})).
		\end{aligned}
	\end{equation}
	
	By ensuring that Equation~(\ref{relacaoy3erros}) holds for all
	$i = 4, \ldots, m$, it follows that the syndrome expressions are valid,
	that is, the symbol errors indeed occurred at positions $l_1$, $l_2$,
	and $l_3$. To determine the magnitudes of the errors, we now solve a
	linear system with three equations involving the variables
	$\textbf{e}^T_{l_1}$, $\textbf{e}^T_{l_2}$, and $\textbf{e}^T_{l_3}$.
	
	The analysis of this case highlights the main structural aspects of the problem and illustrates how the proposed approach can be extended to handle additional errors, although this extension leads to a significant increase in the algebraic complexity of the involved expressions.
	
	The expressions~(\ref{ri1hat}) and~(\ref{ri2hat}) are highly complex to compute when an arbitrary superregular matrix $A$ is used, due to their dependence on multiple indices. Motivated by Corollary~\ref{corollary}, we present simplified expressions when $A$ is a Vandermonde matrix.
	
	\begin{corollary}\label{Vand3erros} If $A$ is a Vandermonde matrix, then the simplified forms of (\ref{ri1hat}) and (\ref{ri2hat}) are
		\begin{equation}\label{barrar41}
			\widehat{r_{i1}}=l_3-Z(l_2-l_3)+Z(2l_2-2l_3)
		\end{equation}
		and
		\begin{equation}\label{barrar42}
			\widehat{r_{i2}} =
			3l_2 - l_3 - Z(l_2 - l_3)
			+ Z\big(l_3 - l_2 - Z(l_2 - l_3) + Z(2l_2 - 2l_3)\big)
		\end{equation}
		for $i=4,\ldots,m$, where $l_2$ and $l_3$ are the second and third error positions.
		
	\end{corollary}
	\begin{proof}  Let $A$ be a Vandermonde matrix. If $i-2$, $i-1$ and $i$ are three consecutive rows and $l_j$ and $l_k$ are any two columns of $A$, then, as in the proof of Corollary \ref{corollary}, it is shown that
		\begin{equation}\label{provalk}
			\begin{aligned}
				l_k &=\sigma(i,l_k)+\sigma(i-1,l_k)+Z\big(\sigma(i,l_j)-\sigma(i,l_k)-\sigma(i-1,l_j)+\sigma(i-1,l_k)\big)\\ & +Z\big(\sigma(i-1,l_j)-\sigma(i-1,l_k)-\sigma(i-2,l_j)+\sigma(i-2,l_k)\big).
			\end{aligned}
		\end{equation}
		And, if $i-3$, $i-2$, $i-1$ and $i$ are four consecutive rows and $l_j$ and $l_k$ are any two columns of $A$, then it is shown that
		\begin{equation}\label{prova2lk}
			\begin{aligned}
				2l_k &=\sigma(i,l_k)+\sigma(i-2,l_k)+Z(\sigma(i,l_j)-\sigma(i,l_k)-\sigma(i-1,l_j)+\sigma(i-1,l_k))\\ & +Z\big(\sigma(i-2,l_j)-\sigma(i-2,l_k)-\sigma(i-3,l_j)+\sigma(i-3,l_k)\big),
			\end{aligned}
		\end{equation}
		since, in a Vandermonde matrix, the differences $\sigma(i,l_k)-\sigma(i-1,l_k)$ are constant and equal to $l_k$, and the Zech logarithm terms cancel pairwise when summing over consecutive rows, yielding the stated identities.
		Using (\ref{provalk}) and (\ref{prova2lk})  in (\ref{ri1}) and (\ref{ri2}), we obtain
		\begin{equation}\label{ri1barra}
			r_{i1}-r_{(i-1)1}=l_2 \ \ \mbox{and} \ \  r_{i2}-r_{(i-1)2}=l_3,
		\end{equation}
		and
		\begin{equation}\label{ri2barra}
			r_{i1}-r_{(i-2)1}=2l_2 \ \ \mbox{and} \ \  r_{i2}-r_{(i-2)2}=2l_3,
		\end{equation}
		for $i=4,\ldots,m$. Now, using (\ref{ri1barra}) and (\ref{ri2barra})  in (\ref{bari1}) and (\ref{bari2}), then 
		\begin{equation}\label{ribarraV}
			\overline{r_{(i-1)1}}=l_2 
			\quad \mbox{and} \quad  
			\overline{r_{(i-1)2}}=r_{(i-1)2}+Z(l_2-l_3),
		\end{equation}
		for $i=4,\ldots,m$. 
		Observe that the dependence on $r_{(i-1)2}$ in $\overline{r_{(i-1)2}}$ disappears after substitution into (\ref{ri1hat}) and (\ref{ri2hat}), since \eqref{ri1barra} allows us to rewrite $r_{(i-1)2}$ in terms of $r_{i2}$ and $l_3$, and the resulting terms are then absorbed into Zech logarithm expressions and cancel, yielding formulas depending only on $l_2$ and $l_3$.
		Finally, substituting (\ref{ribarraV}) into (\ref{ri1hat}) and (\ref{ri2hat}) and simplifying the resulting expressions, we obtain
		\begin{equation}
			\widehat{r_{i1}}=l_3-Z(l_2-l_3)+Z(2l_2-2l_3)
		\end{equation}
		and
		\begin{equation}
			\widehat{r_{i2}} =
			3l_2 - l_3 - Z(l_2 - l_3)
			+ Z\big(l_3 - l_2 - Z(l_2 - l_3) + Z(2l_2 - 2l_3)\big)
		\end{equation}
		for $i=4,\ldots,m$, which proves Corollary \ref{Vand3erros}. These expressions depend only on the error positions $l_2$ and $l_3$.
	\end{proof}

	Presented below is an example of correcting three symbol errors. Note that the decoding steps are the same as those in Algorithm~\ref{algoritmo_dois_erros}, with the adaptations presented above.

	\begin{example} \label{exemplo2_3erros}
		Consider the primitive polynomial $p(x)= x^5 + x^2 +1 \in \mathbb{F}_2[x]$, generating an MDS array code $\mathcal{C}$ with parameters $[11, 5, 7]$ over $\mathbb{F}_2^5$. The superregular Vandermonde matrix $A$ and the Frobenius companion matrix $C$ are given, respectively, by
		\[
		A = \left [
		\begin{array}{c c c c c}
			1 & 1 & 1 & 1 & 1 \\
			\alpha   & \alpha^2    & \alpha^3    & \alpha^4    & \alpha^5\\\
			\alpha^2 & \alpha^4    & \alpha^6    & \alpha^8    & \alpha^{10}\\
			\alpha^3 & \alpha^6    & \alpha^9    & \alpha^{12} & \alpha^{15}\\
			\alpha^4 & \alpha^8    & \alpha^{12} & \alpha^{16} & \alpha^{20}\\
			\alpha^5 & \alpha^{10} & \alpha^{15} & \alpha^{20} & \alpha^{25}\\
		\end{array}
		\right ] \ \ \mbox{and} \ \
		C = 
		\left [
		\begin{array}{c c c c c}
			0 & 0 & 0 & 0 & 1 \\
			1 & 0 & 0 & 0 & 0 \\
			0 & 1 & 0 & 0 & 1 \\
			0 & 0 & 1 & 0 & 0\\
			0 & 0 & 0 & 1 & 0
		\end{array}
		\right ]. 
		\]
		Thus
		\[   
		H = \left [
		\begin{array}{c c c c c | c}
			I_{5} & I_{5} & I_{5} & I_{5} & I_{5}  & \\
			C & C^2 & C^3 & C^4 & C^5  & \\
			C^2 & C^4 & C^6 & C^8 & C^{10}  &  I_{30} \\
			C^3 & C^6 & C^9 & C^{12} & C^{15} &  \\
			C^4 & C^8 & C^{12} & C^{16} & C^{20} & \\
			C^5 & C^{10} & C^{15} & C^{20} & C^{25}  &
		\end{array}
		\right ]
		\]
		is the parity-check matrix of the code $\mathcal{C}$. Assume the vector received is  \[\textbf{v}= [01011 \; 10010 \; 11100 \; 00100 \; 10001 \; 01110 \; 00111 \; 01101 \; 01001 \; 01010 \; 00001].\] The syndromes from (\ref{sindrome}) are given by
		
		\[
		\textbf{s}_1^T = \left [
		\begin{array}{c}
			1 \\
			1 \\
			1 \\
			1 \\
			0 \\
		\end{array}
		\right ],
		\textbf{s}_2^T= \left [ 
		\begin{array}{c}
			0 \\
			0 \\
			1 \\
			0 \\
			0 \\
		\end{array}
		\right ], 
		\textbf{s}_3^T = \left [ 
		\begin{array}{c}
			1 \\
			1 \\
			0 \\
			1 \\
			1 \\
		\end{array}
		\right ],
		\textbf{s}_4^T = \left [ 
		\begin{array}{c}
			0 \\
			0 \\
			1 \\
			0 \\
			0 \\
		\end{array}
		\right ],
		\textbf{s}_5^T = \left [ 
		\begin{array}{c}
			0 \\
			1 \\
			0 \\
			0 \\
			0 \\
		\end{array}
		\right ],
		\textbf{s}_6^T = \left [
		\begin{array}{c}
			0 \\
			1 \\
			0 \\
			0 \\
			1 \\
		\end{array}
		\right ].
		\]
		As all syndromes are nonzero, taking $l_1=1$ we calculate $\textbf{y}_i$ for $i=1,2,\ldots, 6$, as in Equation (\ref{calculoy}),
		\[
		\textbf{y}_1^T = \left [
		\begin{array}{c}
			0 \\
			1 \\
			1 \\
			0 \\
			1 \\
		\end{array}
		\right ],
		\textbf{y}_2^T= \left [ 
		\begin{array}{c}
			0 \\
			1 \\
			0 \\
			1 \\
			1 \\
		\end{array}
		\right ], 
		\textbf{y}_3^T = \left [ 
		\begin{array}{c}
			1 \\
			1 \\
			0 \\
			0 \\
			1 \\
		\end{array}
		\right ],
		\textbf{y}_4^T = \left [ 
		\begin{array}{c}
			1 \\
			1 \\
			1 \\
			0 \\
			1 \\
		\end{array}
		\right ],
		\textbf{y}_5^T = \left [ 
		\begin{array}{c}
			0 \\
			1 \\
			0 \\
			1 \\
			0 \\
		\end{array}
		\right ],
		\textbf{y}_6^T = \left [
		\begin{array}{c}
			0 \\
			1 \\
			1 \\
			0 \\
			1 \\
		\end{array}
		\right ].
		\]
		Note that at this point in the algorithm, we should check for the possibility of a second symbol error combined with $l_1=1$. Following the Algorithm \ref{algoritmo_dois_erros} we suppose a second symbol error at $l_2=2$ and compute $r_{i - 2}=2$ for $i= 3,4,5,6$ by using (\ref{eq:l2}), and then we verify that
		\[\textbf{y}_i^T \neq C^2 \textbf{y}_{i-1}^T,\]
		for all $i=3,4,5,6$. Since there is no verification if we consider all possibilities for $l_2$, we would have to assume that a third symbol error at $l_3=3$ and find the values such that 
		in (\ref{relacaoy3erros}), for $i= 4,5,6$. For (\ref{barrar41}) and (\ref{barrar42}) we compute
		\[
		\begin{split}
			\widehat{r_{i1}}=3-Z(-1)+Z(-2)=3-17+3=-11
		\end{split}
		\]
		and
		\[
		\begin{split}
			\widehat{r_{i2}}
			&= 3\cdot 2 - 3 - Z(-1)
			+ Z\big(3-2 - Z(-1) + Z(4-6)\big) \\
			&= 6-3 - Z(-1)
			+ Z\big(1 - Z(-1) + Z(-2)\big) = 3 - 17 + Z\big(1 - 17 + 3\big) \\
			&= 3 - 17 + Z(-13) = -14 + 19 = 5.
		\end{split}
		\]
		for $i=4,5,6$. Then, in fact
		\[\textbf{y}_i^T=C^{-11} \textbf{y}_{i-1}^T + C^{5} \textbf{y}_{i-2}^T,\]
		for all $i=4,5,6$. So, the position of the three symbol errors are $l_1=1$, $l_2=2$ and $l_3=3$. Solving the linear system we obtain the magnitude of the errors
		\[
		\textbf{e}_1=    [11010], \ \ \ \textbf{e}_2=[01010] \ \ \ \mbox{and} \ \ \ \textbf{e}_3=  	[01110].
		\]
		Therefore, the codeword is
		\[
		\begin{aligned}
			\textbf{c} & = [01011 \;\; 10010 \;\; 11100 \;\; 00100 \;\; 10001 \;\; 01110 \;\; 00111 \;\; 01101 \;\; 01001 \;\; 01010 \;\; 00001] \\
			& - [\textbf{11010} \; \textbf{01010} \; \textbf{01110} \;\; 00000 \;\; 00000 \;\; 00000 \;\; 00000 \;\; 00000 \;\; 00000 \;\; 00000 \;\; 00000] \\
			& = [10001 \;\; 11000 \;\; 10010 \;\; 00100 \;\; 10001 \;\; 01110 \;\; 00111 \;\; 01101 \;\; 01001 \;\; 01010 \;\; 00001].
		\end{aligned}
		\]
		
	\end{example}
	
	\section{Conclusion}\label{conclusao}
	
	In this paper, we investigated decoding algorithms for MDS array codes over $\mathbb{F}_q^b$ constructed from superregular matrices. We presented explicit decoding procedures capable of correcting one and two symbol errors without prior knowledge of their locations. These algorithms apply to a broad class of MDS array codes and rely on algebraic properties of superregular matrices. A relevant contribution is the simplification of the decoding procedure when Vandermonde superregular matrices are used, where closed-form expressions are obtained for the two-error scenario, reducing the number of algebraic operations and highlighting the practical relevance of such matrices in efficient decoding.
	
	The three-symbol-error case was examined as a partial extension of the proposed framework. We focused on the most algebraically demanding configuration, namely the case in which all three errors occur in information symbols. The remaining configurations can be treated by arguments analogous to those used in the two-error case and were therefore omitted. This analysis provides insight into the structural behavior of the decoding process for higher numbers of errors. Although the approach can, in principle, be extended to correct up to $\lfloor (d-1)/2 \rfloor$ symbol errors, this leads to a significant increase in algebraic complexity.
	
	From an application perspective, the proposed approach provides a flexible alternative to RAID~6 schemes, supporting general parameter configurations and enabling the correction of multiple symbol errors without location information. This flexibility, however, comes at the cost of increased algebraic complexity.
	
	Future work includes extending the decoding framework to handle a larger number of symbol errors, improving computational efficiency, and evaluating performance in distributed storage systems and erasure channels. In particular, a detailed comparison with classical decoding methods, such as those for Reed--Solomon codes, and the exploration of other families of superregular and block superregular matrices may provide further insight into the practical applicability of the proposed decoding algorithms.\\

	\textbf{Acknowledgments:}
	The authors gratefully acknowledge the financial support provided by the
	São Paulo Research Foundation (FAPESP) under grants 2024/05051-7, 2024/00923-6, 2017/17948-8, and 2013/25977-7, the Funding Authority for
	Studies and Projects (FINEP) under grant
	0527/18, and the National Council for Scientific
	and Technological Development (CNPq) under grant 405842/2023-6.
	
	\bibliographystyle{plain}
	
	\bibliography{sn-bibliography}
	
\end{document}